\documentclass[%
superscriptaddress,
nofootinbib,
nobibnotes,
 amsmath,amssymb,
 aps,
floatfix,
prl,
twocolumn,
longbibliography
]{revtex4-2}

\usepackage{qcircuit}
\usepackage{float}
\usepackage{mathtools}
\usepackage{hyperref}
\usepackage{fancyhdr}
\usepackage{extramarks}
\usepackage{amsmath}
\usepackage{amsthm}
\usepackage{amsfonts}
\usepackage{amssymb}
\usepackage{tikz}
\usepackage{dsfont}
\usepackage{parskip}

\usepackage{algorithm}
\usepackage{algpseudocode}
\usepackage{algorithmicx}
\usepackage{bbm}
\usepackage{diagbox}
\usepackage{caption}
\usepackage{subcaption}
\usepackage{wasysym}

\usepackage{amssymb}
\usepackage{verbatim}
\usepackage{tabularx,lipsum,environ}
\usepackage{csquotes}

\newtheorem{theorem}{Theorem}
\newtheorem{corollary}[theorem]{Corollary}
\newtheorem{lemma}[theorem]{Lemma}

\newcommand{\prob}[1]{\mathbb{P}\left(#1\right)}

\newcommand{\esp}[1]{\mathbb{E}\left[#1\right]}

\newcommand{\floor}[1]{\left\lfloor #1 \right\rfloor}

\usepackage{appendix}

\usepackage{tikzit}

\tikzstyle{box}=[shape=rectangle, text height=1.5ex, text depth=0.25ex, yshift=0.5mm, fill=white, draw=black, minimum height=5mm, yshift=-0.5mm, minimum width=5mm, font={\small}]
\tikzstyle{gate}=[shape=rectangle, text height=1.5ex, text depth=0.25ex, yshift=0.5mm, fill=white, draw=black, minimum height=5mm, yshift=-0.5mm, minimum width=5mm, font={\small}, tikzit category=circuit]
\tikzstyle{big gate}=[shape=rectangle, text height=1.5ex, text depth=0.25ex, yshift=0.5mm, fill=white, draw=black, minimum height=10mm, yshift=-0.5mm, minimum width=5mm, font={\small}, tikzit category=circuit]
\tikzstyle{Z dot}=[inner sep=0mm, minimum size=2mm, shape=circle, draw=black, fill={rgb,255: red,221; green,255; blue,221}, tikzit category=zx]
\tikzstyle{Z phase dot}=[minimum size=5mm, font={\footnotesize\boldmath}, shape=rectangle, rounded corners=2mm, inner sep=0.2mm, outer sep=-2mm, scale=0.8, tikzit shape=circle, draw=black, fill={rgb,255: red,221; green,255; blue,221}, tikzit draw=blue, tikzit category=zx]
\tikzstyle{X dot}=[Z dot, shape=circle, draw=black, fill={rgb,255: red,255; green,136; blue,136}, tikzit category=zx]
\tikzstyle{X phase dot}=[Z phase dot, tikzit shape=circle, tikzit draw=blue, fill={rgb,255: red,255; green,136; blue,136}, font={\footnotesize\boldmath}, tikzit category=zx]
\tikzstyle{hadamard}=[fill=yellow, draw=black, shape=rectangle, inner sep=0.6mm, minimum height=1.5mm, minimum width=1.5mm, tikzit category=zx]
\tikzstyle{paulibox}=[fill={rgb,255: red,221; green,221; blue,255}, draw=black, shape=rectangle, inner sep=0.6mm, minimum height=5mm, minimum width=5mm, font={\footnotesize}, text height=1.5ex, text depth=0.25ex, tikzit category=zx]
\tikzstyle{vertex}=[inner sep=0mm, minimum size=1mm, shape=circle, draw=black, fill=black, tikzit category=misc]
\tikzstyle{vertex set}=[inner sep=0mm, minimum size=1mm, shape=circle, draw=black, fill=white, font={\footnotesize\boldmath}, tikzit category=misc]
\tikzstyle{small black dot}=[fill=black, draw=black, shape=circle, inner sep=0pt, minimum width=1.2mm, tikzit category=circuit]
\tikzstyle{cnot ctrl}=[fill=black, draw=black, shape=circle, inner sep=0pt, minimum width=1.2mm, tikzit category=circuit]
\tikzstyle{cnot targ}=[fill=white, draw=white, shape=circle, tikzit category=circuit, label={center:$\oplus$}, inner sep=0pt, minimum width=2.1mm, tikzit fill={rgb,255: red,102; green,204; blue,255}, tikzit draw=black]
\tikzstyle{ket}=[fill=white, draw=black, shape=regular polygon, regular polygon sides=3, regular polygon rotate=-30, scale=0.7, inner sep=1pt, tikzit category=circuit, tikzit shape=rectangle, tikzit fill=green]
\tikzstyle{bra}=[fill=white, draw=black, shape=regular polygon, regular polygon sides=3, regular polygon rotate=30, scale=0.7, inner sep=1pt, tikzit category=circuit, tikzit shape=rectangle, tikzit fill=red]
\tikzstyle{scalar}=[shape=rectangle, text height=1.5ex, text depth=0.25ex, yshift=0.5mm, fill=white, draw=black, minimum height=5mm, yshift=-0.5mm, minimum width=5mm, font={\small}]
\tikzstyle{clabel}=[fill=white, draw=none, shape=rectangle, tikzit fill={rgb,255: red,56; green,255; blue,242}, font={\footnotesize}, inner sep=1pt, tikzit category=labels]
\tikzstyle{empty diagram}=[draw={gray!40!white}, dashed, shape=rectangle, minimum width=1cm, minimum height=1cm, tikzit category=misc]

\tikzstyle{simple}=[-]
\tikzstyle{hadamard edge}=[-, dashed, dash pattern=on 2pt off 0.5pt, thick, draw={rgb,255: red,68; green,136; blue,255}]
\tikzstyle{box edge}=[-, dashed, dash pattern=on 2pt off 0.5pt, thick, draw={rgb,255: red,203; green,192; blue,225}]
\tikzstyle{brace edge}=[-, tikzit draw=blue, decorate, decoration={brace,amplitude=1mm,raise=-1mm}]
\tikzstyle{diredge}=[->]
\tikzstyle{double edge}=[-, double, shorten <=-1mm, shorten >=-1mm, double distance=2pt]
\tikzstyle{gray edge}=[-, {gray!60!white}]
\tikzstyle{pointer edge}=[->, very thick, gray]
\tikzstyle{boldedge}=[-, line width=1.6pt, shorten <=-0.17mm, shorten >=-0.17mm]

\tikzstyle{white dot}=[Z dot]
\tikzstyle{gray dot}=[X dot]
\tikzstyle{white phase dot}=[Z phase dot]
\tikzstyle{gray phase dot}=[X phase dot]
\tikzstyle{small hadamard}=[hadamard]
\input{quantum.tikzdefs}

\begin{document}


\title{Classically Simulating Quantum Supremacy IQP Circuits\\ 
through a Random Graph Approach}

\author{Julien Codsi}%
\email{julien.codsi@umontreal.ca}
 \affiliation{University of Montréal}
\author{John van de Wetering}%
 \email{john@vdwetering.name}
 \homepage{http://vdwetering.name}
\affiliation{%
 University of Amsterdam
}%

\date{\today}

\begin{abstract}
Quantum Supremacy is a demonstration of a computation by a quantum computer that can not be performed by the best classical computer in a reasonable time. A well-studied approach to demonstrating this on near-term quantum computers is to use random circuit sampling. It has been suggested that a good candidate for demonstrating quantum supremacy with random circuit sampling is to use \emph{IQP circuits}. These are quantum circuits where the unitary it implements is diagonal.
In this paper we introduce improved techniques for classically simulating random IQP circuits. We find a simple algorithm to calculate an amplitude of an $n$-qubit IQP circuit with dense random two-qubit interactions in time $O(\frac{\log^2 n}{n} 2^n )$, which for sparse circuits (where each qubit interacts with $O(\log n)$ other qubits) runs in $o(2^n/\text{poly}(n))$ for any given polynomial. 
Using a more complicated stabiliser decomposition approach we improve the algorithm for dense circuits to $O\left(\frac{(\log n)^{4-\beta}}{n^{2-\beta}} 2^n \right)$ where $\beta \approx 0.396$.
We benchmarked our algorithm and found that we can simulate up to 50-qubit circuits in a couple of minutes on a laptop. We estimate that 70-qubit circuits are within reach for a large computing cluster.
\end{abstract}

\maketitle

\setlength{\parskip}{1pt plus1pt minus1pt}
\setlength{\parindent}{15pt}

Recent years have seen the development of noisy quantum computers that have enough qubits and coherence to start to probe the limits of classical simulation.
In fact, in 2019 Arute et al.~\cite{arute2019quantum} already claimed to have reached \emph{quantum supremacy}: a quantum computation that cannot be simulated by any classical computer in a reasonable time frame. This was done by sampling from a random quantum circuit, and computing a metric called the linear cross-entropy benchmark (XEB). 
Their claim was that it would take the best supercomputer in the world 10.000 years to simulate the computation they did.
However, soon after that, improvements in tensor contraction techniques reduced this number to just days~\cite{huangClassicalSimulationQuantum2020}, and even hours on a moderately sized GPU cluster~\cite{PhysRevLett.129.090502}. 
By allowing the simulation to produce correlated bitstrings, much higher XEB scores can be reached with fewer resources~\cite{pan2021simulating}, and bypassing directly simulating the computation entirely, non-trivial XEB scores turned out to also be generatable in mere seconds on a single GPU~\cite{gao2021limitations}.

This progress shows that claims of quantum supremacy should be made carefully, as improvements in classical algorithms can quickly gain orders of magnitudes in improvement. In this paper we will consider the classical simulation of a different type of random quantum circuit that has been proposed as a good candidate for quantum supremacy experiments.

Instantaneous Quantum Polynomial (IQP) circuits are quantum circuits where the input is prepared in the all-zero state $\ket{0\cdots 0}$, the unitary is of the form $H^{\otimes n} D H^{\otimes n}$ where $H^{\otimes n}$ is a Hadamard gate applied to all the qubits, and $D$ is a unitary consisting of polynomially many diagonal gates~\cite{shepherd2009temporally}. The name `instantaneous' comes from the fact that all the diagonal gates commute, so that there is no time order encoded into the circuit.
IQP circuits were originally introduced in~\cite{shepherd2009temporally} as a simplified model of quantum computation where interesting, and hard to classically simulate problems could be formulated.
Indeed, it was proven in~\cite{bremner_classical_2011} that the ability to efficiently simulate IQP circuits would imply a collapse of the polynomial hierarchy to the third level, which is considered very unlikely. 
This was improved in~\cite{bremner_average-case_2016} to hardness under a more reasonable additive error bound. 
Then in~\cite{bremner_achieving_2017}, it was shown that even random IQP circuits consisting of just powers of the $T = \text{diag}(1,e^{i\pi/4})$ gate and $O(n\log n)$ $CS = \text{diag}(1,1,1,i)$ gates are likely to be hard to simulate, and that furthermore they can be compiled onto a 2D architecture within a reasonable depth, and that they can be constructed in such a way to be resilient to some noise. These properties make these circuits an interesting candidate for quantum supremacy experiments, and raise the question of where the boundary of classical simulability lies: even though the simulation is likely to be asymptotically hard, it might still be that in practical regimes, the results can still be efficiently simulated.

In this paper we find better algorithms for simulating random $\{T,CS\}$ IQP circuits.  We do this by realising that such circuits follow the structure of Erdös-Rényi random graphs. Such graphs have relatively large independent vertex sets. This allows us to use techniques from the stabiliser decomposition technique of simulation~\cite{bravyiImprovedClassicalSimulation2016,Bravyi2019simulationofquantum,qassim2021improved,kissinger2021simulating,kissinger2022classical} to cut the circuit into a sum of smaller instances. In particular, we find we can exactly calculate amplitudes of random dense Clifford+$T$ IQP circuits in time $O(\frac{\log^2 n}{n} 2^n)$, and with a more complicated algorithm in time $O\left(\frac{(\log n)^{4-\beta}}{n^{2-\beta}} 2^n \right)$ where $\beta \approx 0.396$ is the stabiliser decomposition constant of~\cite{qassim2021improved,kissinger2022classical}. For the random sparse circuits of~\cite{bremner_achieving_2017} we find we can calculate an amplitude in time $O\left(\frac{n\log\log(n)}{\log(n)} 2^{n\left(1- O\left(\frac{\log\log(n)}{\log(n)}\right)\right)} \right)$. Note that this bound is faster than $O(2^n/poly(n))$ for any given polynomial.
We can boost the calculation of amplitudes to a procedure for sampling from the circuit, by using the `gate-by-gate' simulation technique of~\cite{bravyi2021simulate} that avoids calculating marginals. This technique turns out to be particularly suited to IQP circuits, as it only requires an additional sample per non-diagonal gate, of which there are $O(n)$ (corresponding to the layers of Hadamard gates). Our algorithms can hence weakly sample from the dense, respectively sparse, circuits in time $O\left(\frac{(\log n)^{4-\beta}}{n^{1-\beta}} 2^n \right)$, respectively $O\left(\frac{n^2\log\log(n)}{\log(n)} 2^{n\left(1- O\left(\frac{\log\log(n)}{\log(n)}\right)\right)} \right)$.
This should be compared to the cost $O(n^2 2^n)$, respectively $O(n\log n 2^n)$ of doing a state vector simulation, or $O(n^4 2^{O(n^2)})$, respectively $O(n^2\log^2 n 2^{O(n\log n)})$ of using stabiliser decompositions directly.


Our results give asymptotic polynomial improvements over the previous best, but these improvements are also of practical significance. 
We implemented the simpler algorithm, and found that, depending on the density of the circuit, we can calculate an amplitude of 30- to 50-qubit IQP circuits on a single CPU core on a laptop in a couple of minutes; see Figure~\ref{fig:Benchmarks}. As our algorithm is easily distributed in parallel, we estimate that a 100.000 CPU core cluster could calculate an amplitude of a dense 60-qubit circuit in about an hour, and that 70-qubit circuits should be within reach of the world's best supercomputers.

\emph{The ZX-calculus.}---%
Since our algorithm was found by representing IQP computations as ZX-diagrams, we will give a brief overview of the ZX-calculus~\cite{CD1,CD2}. For an in-depth
reference see~\cite{vandewetering2020zxcalculus}.
The \zxcalculus is a diagrammatic language similar to the familiar
quantum circuit notation.  A \emph{\zxdiagram} (or simply
\emph{diagram}) consists of \emph{wires} and \emph{spiders}.  Wires
entering the diagram from the left are \emph{inputs}; wires exiting to
the right are \emph{outputs}.  Given two diagrams we can compose them
by joining the outputs of the first to the inputs of the second, or
form their tensor product by simply stacking the two diagrams.

Spiders are linear operations which can have any number of input or output
wires.  There are two varieties, $Z$ spiders depicted as green dots: 
\begin{equation}
\begin{tikzpicture}
	\begin{pgfonlayer}{nodelayer}
		\node [style=Z phase dot] (0) at (0, 0) {$\alpha$};
		\node [style=none] (1) at (1.25, 0.825) {};
		\node [style=none] (2) at (-1.25, 0.825) {};
		\node [style=none] (3) at (-1.25, -0.75) {};
		\node [style=none] (4) at (1.25, -0.75) {};
		\node [style=none] (5) at (1.25, 0.5) {};
		\node [style=none] (6) at (-1.25, 0.5) {};
		\node [style=none, rotate=90] (7) at (-1, -0.125) {...};
		\node [style=none, rotate=90] (8) at (1, -0.125) {...};
	\end{pgfonlayer}
	\begin{pgfonlayer}{edgelayer}
		\draw [in=-141, out=0, looseness=0.75] (3.center) to (0);
		\draw [in=180, out=-39, looseness=0.75] (0) to (4.center);
		\draw [in=180, out=22, looseness=0.75] (0) to (5.center);
		\draw [in=180, out=39, looseness=0.75] (0) to (1.center);
		\draw [in=0, out=158, looseness=0.75] (0) to (6.center);
		\draw [in=141, out=0, looseness=0.75] (2.center) to (0);
	\end{pgfonlayer}
\end{tikzpicture} \ := \ \ketbra{0...0}{0...0} +
e^{i \alpha} \ketbra{1...1}{1...1} \label{eq:def-Z-spider}
\end{equation}
and $X$ spiders depicted as red dots:
\begin{equation}
\begin{tikzpicture}
	\begin{pgfonlayer}{nodelayer}
		\node [style=X phase dot] (9) at (0, 0) {$\alpha$};
		\node [style=none] (10) at (1.25, 0.825) {};
		\node [style=none] (11) at (-1.25, 0.825) {};
		\node [style=none] (12) at (-1.25, -0.75) {};
		\node [style=none] (13) at (1.25, -0.75) {};
		\node [style=none] (14) at (1.25, 0.5) {};
		\node [style=none] (15) at (-1.25, 0.5) {};
		\node [style=none, rotate=90] (16) at (-1, -0.125) {...};
		\node [style=none, rotate=90] (17) at (1, -0.125) {...};
	\end{pgfonlayer}
	\begin{pgfonlayer}{edgelayer}
		\draw [in=-141, out=0, looseness=0.75] (12.center) to (9);
		\draw [in=180, out=-39, looseness=0.75] (9) to (13.center);
		\draw [in=180, out=22, looseness=0.75] (9) to (14.center);
		\draw [in=180, out=39, looseness=0.75] (9) to (10.center);
		\draw [in=0, out=158, looseness=0.75] (9) to (15.center);
		\draw [in=141, out=0, looseness=0.75] (11.center) to (9);
	\end{pgfonlayer}
\end{tikzpicture} \ := \ \ketbra{+...+}{+...+} +
e^{i \alpha} \ketbra{-...-}{-...-}
\end{equation}
When the phase $\alpha$ is zero, we will omit it from the notation.
The diagram as a whole corresponds to a linear map built from the
spiders (and permutations) by the usual composition and tensor product
of linear maps.  As a special case, diagrams with no inputs represent
(unnormalised) state preparations.
For instance:
\begin{equation}\label{eq:state-ZX}
\begin{array}{rcccl}
  \begin{tikzpicture}
	\begin{pgfonlayer}{nodelayer}
		\node [style=Z dot] (0) at (-0.5, 0) {};
		\node [style=none] (1) at (0.5, 0) {};
	\end{pgfonlayer}
	\begin{pgfonlayer}{edgelayer}
		\draw (0) to (1.center);
	\end{pgfonlayer}
\end{tikzpicture} & = & \ket{0} + \ket{1}& \ =& \sqrt{2}\ket{+} \\[0.2cm]
  \begin{tikzpicture}
	\begin{pgfonlayer}{nodelayer}
		\node [style=X dot] (0) at (-0.5, 0) {};
		\node [style=none] (1) at (0.5, 0) {};
	\end{pgfonlayer}
	\begin{pgfonlayer}{edgelayer}
		\draw (0) to (1.center);
	\end{pgfonlayer}
\end{tikzpicture} & = & \ket{+} + \ket{-}& \ =& \sqrt{2}\ket{0} \\[0.2cm]
  \begin{tikzpicture}
	\begin{pgfonlayer}{nodelayer}
		\node [style=Z phase dot] (0) at (0, 0) {$\alpha$};
		\node [style=none] (1) at (-1.25, 0) {};
		\node [style=none] (2) at (1.25, 0) {};
	\end{pgfonlayer}
	\begin{pgfonlayer}{edgelayer}
		\draw (1.center) to (0);
		\draw (0) to (2.center);
	\end{pgfonlayer}
\end{tikzpicture} & = & \ketbra{0}{0} + e^{i \alpha} \ketbra{1}{1}&\ = & Z_\alpha \\[0.2cm]
  \end{array}
\end{equation}
Here the last one is the $Z_\alpha$ phase gate.

For convenience, special notation for
the Hadamard gate is used:
\begin{equation}\label{eq:Hdef}
\hfill
\begin{tikzpicture}
	\begin{pgfonlayer}{nodelayer}
		\node [style=none] (0) at (7.3, 0) {};
		\node [style=none] (1) at (9.3, 0) {};
		\node [style=hadamard] (2) at (8.3, 0) {};
		\node [style=none] (3) at (6.3, 0) {$=:$};
		\node [style=none] (4) at (2.05, 0) {};
		\node [style=Z phase dot] (5) at (2.8, 0) {$\frac\pi2$};
		\node [style=X phase dot] (6) at (3.8, 0) {$\frac\pi2$};
		\node [style=Z phase dot] (7) at (4.8, 0) {$\frac\pi2$};
		\node [style=none] (8) at (5.55, 0) {};
		\node [style=box] (9) at (-3, 0) {H};
		\node [style=none] (10) at (-4.5, 0) {};
		\node [style=none] (11) at (-1.75, 0) {};
		\node [style=none] (12) at (-0.75, 0) {=};
		\node [style=none] (13) at (1, 0.1) {$e^{-i\pi/4}$};
	\end{pgfonlayer}
	\begin{pgfonlayer}{edgelayer}
		\draw (0.center) to (2);
		\draw (2) to (1.center);
		\draw (4.center) to (8.center);
		\draw (10.center) to (11.center);
	\end{pgfonlayer}
\end{tikzpicture}
\hfill
\end{equation}

Two diagrams are considered \emph{equal} when one can be deformed to
the other by moving the vertices around in the plane, bending,
unbending, crossing, and uncrossing wires, as long as the connectivity
and the order of the inputs and outputs is maintained. Equivalently, a
ZX-diagram can be considered as a graphical depiction of a tensor network,
as in e.g.~\cite{Penrose}. The Z- and X-spiders are symmetric tensors, and hence, like for other tensor networks of symmetric tensors, the interpretation of a ZX-diagram is unaffected by deformation. 

Quantum circuits can be translated into \zxdiagrams in a straightforward manner. 
The controlled $Z$ phase gates, Hadamard and the $Z$ phase gates each have a simple representation as a ZX-diagram:
\begin{equation}\label{eq:zx-gates}
\CZ_\alpha = \sqrt{2}\begin{tikzpicture}
	\begin{pgfonlayer}{nodelayer}
		\node [style=none] (0) at (-1, 1) {};
		\node [style=none] (1) at (2, 1) {};
		\node [style=none] (2) at (-1, -1) {};
		\node [style=none] (3) at (2, -1) {};
		\node [style=Z phase dot] (4) at (0.25, 1) {$\frac\alpha2$};
		\node [style=Z phase dot] (5) at (0.25, -1) {$\frac\alpha2$};
		\node [style=X dot] (6) at (0.25, 0) {};
		\node [style=Z phase dot] (7) at (1.625, 0) {$\ -\frac\alpha2~$};
	\end{pgfonlayer}
	\begin{pgfonlayer}{edgelayer}
		\draw (0.center) to (1.center);
		\draw (3.center) to (2.center);
		\draw (4) to (6);
		\draw (6) to (5);
		\draw (6) to (7);
	\end{pgfonlayer}
\end{tikzpicture} \qquad H = \begin{tikzpicture}
	\begin{pgfonlayer}{nodelayer}
		\node [style=none] (0) at (-0.25, 0) {};
		\node [style=none] (1) at (1.75, 0) {};
		\node [style=hadamard] (2) at (0.75, 0) {};
	\end{pgfonlayer}
	\begin{pgfonlayer}{edgelayer}
		\draw (0.center) to (2);
		\draw (2) to (1.center);
	\end{pgfonlayer}
\end{tikzpicture}
\qquad Z_\alpha = \begin{tikzpicture}
	\begin{pgfonlayer}{nodelayer}
		\node [style=Z phase dot] (0) at (0, 0) {$\alpha$};
		\node [style=none] (1) at (-1.25, 0) {};
		\node [style=none] (2) at (1.25, 0) {};
	\end{pgfonlayer}
	\begin{pgfonlayer}{edgelayer}
		\draw (1.center) to (0);
		\draw (0) to (2.center);
	\end{pgfonlayer}
\end{tikzpicture}
\end{equation}
Here, the way we represent the $\CZ_\alpha$ gates is as a \emph{phase gadget}~\cite{kissinger2019tcount}, a particularly useful type of subdiagram that feature heavily in ZX-calculus based optimisation routines~\cite{kissinger2019tcount,deBeaudrap2020Techniques,wetering-gflow}.
Note that we will be particularly interested in the $CS=\CZ_{\frac\pi2}$ gate and the $T=Z_{\frac\pi4}$ gate.
Since the gates of Eq.~\eqref{eq:zx-gates} form a universal gate set, by composing them we can represent any quantum circuit as a ZX-diagram.
In fact, as we can also represent state preparations and post-selections, 
ZX-diagrams with arbitrary angles are expressive enough to represent any linear map~\cite{CD2}. When we restrict the angles to multiples of $\pi/2$, the maps it represents correspond to Clifford maps: linear maps that can be expressed as a combination of stabiliser state preparations, Clifford unitaries, and stabiliser post-selections~\cite{BackensCompleteness}. 
Instead restricting the angles to multiples of $\pi/4$ gives us the Clifford+$T$ \emph{fragment}, which corresponds to those linear maps that can be constructed from Clifford+$T$ unitaries together with state preparations and post-selections~\cite{ng_completeness_2018,SimonCompleteness}.

In addition to this extra flexibility which allows us to represent arbitrary linear maps, the real utility for ZX-diagrams comes from the set of rewrite rules they satisfy. This set of equations is called the \zxcalculus. Diagrams that can be transformed into each other using the rules of the ZX-calculus correspond to equal linear maps. We will only need a small number of rules:
\begin{equation}\label{eq:ZX-rules}
    \begin{tikzpicture}
	\begin{pgfonlayer}{nodelayer}
		\node [style=none] (24) at (-4.5, 2.5) {$=$};
		\node [style=none] (27) at (-1, 3.25) {};
		\node [style=none, rotate=90] (28) at (-3.375, 2.5) {...};
		\node [style=none] (29) at (-1, 1.625) {};
		\node [style=none, rotate=90] (30) at (-1.1, 2.5) {...};
		\node [style=none] (31) at (-3.5, 3.25) {};
		\node [style=Z phase dot] (32) at (-2.25, 2.5) {$\ \alpha\!+\!\beta\ $};
		\node [style=none] (34) at (-3.5, 1.625) {};
		\node [style=X dot] (35) at (1.125, 0.5) {};
		\node [style=none] (36) at (3, 0.5) {};
		\node [style=hadamard] (37) at (2.125, 0.5) {};
		\node [style=none] (38) at (4, 0.5) {=};
		\node [style=Z dot] (39) at (5.125, 0.5) {};
		\node [style=none] (40) at (6.125, 0.5) {};
		\node [style=Z phase dot] (41) at (-8.35, 0.5) {$\alpha$};
		\node [style=none] (42) at (-6.125, 0.5) {};
		\node [style=X phase dot] (43) at (-7.175, 0.5) {$b\pi$};
		\node [style=none] (44) at (-5.25, 0.5) {=};
		\node [style=Z phase dot] (45) at (-2, 0.5) {$\ (-1)^b\alpha~$};
		\node [style=none] (46) at (-0.25, 0.5) {};
		\node [style=none] (48) at (-4, 0.625) {$e^{ib\alpha}$};
		\node [style=none] (49) at (-2.75, -0.625) {$b\in\{0,1\}$};
		\node [style=none] (50) at (-5.5, 3.25) {};
		\node [style=none, rotate=90] (51) at (-8.5, 2.5) {...};
		\node [style=none] (52) at (-5.5, 1.625) {};
		\node [style=none, rotate=90] (53) at (-5.75, 2.5) {...};
		\node [style=none] (54) at (-8.75, 3.25) {};
		\node [style=Z phase dot] (55) at (-7.75, 2.5) {$\alpha$};
		\node [style=none] (56) at (-8.75, 1.625) {};
		\node [style=Z phase dot] (57) at (-6.5, 2.5) {$\beta$};
		\node [style=X dot] (58) at (2.25, 2.5) {};
		\node [style=none] (59) at (3, 1.75) {};
		\node [style=none] (60) at (3, 3.25) {};
		\node [style=Z phase dot] (61) at (1.125, 2.5) {$k\pi$};
		\node [style=none] (62) at (4, 2.5) {=};
		\node [style=none] (63) at (7.375, 1.875) {};
		\node [style=none] (64) at (7.375, 3.25) {};
		\node [style=Z phase dot] (65) at (6.375, 3.25) {$k\pi$};
		\node [style=Z phase dot] (66) at (6.375, 1.875) {$k\pi$};
		\node [style=none] (67) at (5.1, 2.5) {$\frac{1}{\sqrt{2}}$};
	\end{pgfonlayer}
	\begin{pgfonlayer}{edgelayer}
		\draw [style=simple, in=-120, out=0, looseness=0.75] (34.center) to (32);
		\draw [style=simple, in=-180, out=-60, looseness=0.75] (32) to (29.center);
		\draw [style=simple, in=180, out=45, looseness=0.75] (32) to (27.center);
		\draw [style=simple, in=0, out=135, looseness=0.75] (32) to (31.center);
		\draw (35) to (36.center);
		\draw (39) to (40.center);
		\draw (41) to (42.center);
		\draw (45) to (46.center);
		\draw [style=simple, in=-120, out=0, looseness=0.75] (56.center) to (55);
		\draw [style=simple, in=0, out=135, looseness=0.75] (55) to (54.center);
		\draw [in=180, out=45, looseness=0.75] (57) to (50.center);
		\draw [in=-180, out=-60, looseness=0.75] (57) to (52.center);
		\draw (55) to (57);
		\draw [in=-60, out=180] (59.center) to (58);
		\draw [in=60, out=180] (60.center) to (58);
		\draw (61) to (58);
		\draw (65) to (64.center);
		\draw (66) to (63.center);
	\end{pgfonlayer}
\end{tikzpicture}
\end{equation}
These are the \emph{spider-fusion} rule---that adjacent spiders of the same colour fuse together (which also holds for the X-spider)---and special cases of the \emph{colour-change} rule---that a Hadamard can be commuted through a spider to change its colour---and the \emph{$\pi$-copy} rules---that a $\pi$ phase can be commuted through the opposite colour~\cite{vandewetering2020zxcalculus}.

\emph{IQP circuits as ZX-diagrams.}---%
As a ZX-diagram, an IQP circuit can be represented, up to some known global non-zero scalar as
\begin{equation}
  \scalebox{0.9}{\begin{tikzpicture}
	\begin{pgfonlayer}{nodelayer}
		\node [style=none] (0) at (-6.5, 3) {};
		\node [style=none] (1) at (-6.5, 1) {};
		\node [style=none] (3) at (-6.5, -3) {};
		\node [style=none] (4) at (-5.75, -1) {$\vdots$};
		\node [style=Z phase dot] (5) at (-4, 3) {$\ x_1\frac{\pi}{4}~$};
		\node [style=Z phase dot] (9) at (-4, 1) {$\ x_2\frac{\pi}{4}~$};
		\node [style=Z phase dot] (11) at (-4, -3) {$\ x_n\frac{\pi}{4}~$};
		\node [style=none] (12) at (2.75, 3) {};
		\node [style=none] (13) at (2.75, 1) {};
		\node [style=none] (14) at (2.75, -3) {};
		\node [style=none] (15) at (2, -1) {$\vdots$};
		\node [style=X dot] (16) at (-2, 2) {};
		\node [style=Z phase dot] (17) at (0.25, 2) {$\ y_{1,2}\frac{ \pi}{4}~$};
		\node [style=X dot] (18) at (-2, -2) {};
		\node [style=Z phase dot] (19) at (0.25, -2) {$\ y_{1,n}\frac{ \pi}{4}~$};
		\node [style=X dot] (20) at (-2, -0.25) {};
		\node [style=Z phase dot] (21) at (0.25, -0.25) {$\ y_{2,n}\frac{ \pi}{4}~$};
		\node [style=none] (22) at (-4, -1) {$\vdots$};
		\node [style=hadamard] (23) at (-5.75, 3) {};
		\node [style=hadamard] (24) at (-5.75, 1) {};
		\node [style=hadamard] (25) at (-5.75, -3) {};
		\node [style=hadamard] (26) at (2, 3) {};
		\node [style=hadamard] (27) at (2, 1) {};
		\node [style=hadamard] (28) at (2, -3) {};
	\end{pgfonlayer}
	\begin{pgfonlayer}{edgelayer}
		\draw [in=180, out=0] (0.center) to (5);
		\draw (1.center) to (9);
		\draw (3.center) to (11);
		\draw (5) to (12.center);
		\draw (13.center) to (9);
		\draw (11) to (14.center);
		\draw (5) to (16);
		\draw (9) to (16);
		\draw (16) to (17);
		\draw (18) to (19);
		\draw [in=105, out=-60, looseness=0.75] (5) to (18);
		\draw (18) to (11);
		\draw (20) to (21);
		\draw (9) to (20);
		\draw (20) to (11);
	\end{pgfonlayer}
\end{tikzpicture}}
\end{equation}
Where the phases $x_i, y_{i,j} \in \{0,\ldots,7\}$ arise from the powers of the $T$ and $CS$ gates in the circuit. 
We note that $y_{i,j} = 4 $ corresponds to having four CS-gates in a row between the qubit $i$ and $j$ which is equivalent to the identity. We can see this in the ZX-diagram as:
\begin{equation}
    \begin{tikzpicture}
	\begin{pgfonlayer}{nodelayer}
		\node [style=none] (0) at (-8.5, 1) {};
		\node [style=none] (1) at (-8.5, -0.5) {};
		\node [style=Z phase dot] (2) at (-7, 1) {$\ x_i\frac{\pi}{4}~$};
		\node [style=Z phase dot] (3) at (-7, -0.5) {$\ x_j\frac{\pi}{4}~$};
		\node [style=none] (4) at (-4, 1) {};
		\node [style=none] (5) at (-4, -0.5) {};
		\node [style=X dot] (6) at (-5.75, 0.25) {};
		\node [style=Z phase dot] (7) at (-4.75, 0.25) {$\pi$};
		\node [style=none] (8) at (-2, 1) {};
		\node [style=none] (9) at (-2, -0.5) {};
		\node [style=Z phase dot] (10) at (-0.75, 1) {$\ x_i\frac{\pi}{4}~$};
		\node [style=Z phase dot] (11) at (-0.75, -0.5) {$\ x_j\frac{\pi}{4}~$};
		\node [style=none] (12) at (2.25, 1) {};
		\node [style=none] (13) at (2.25, -0.5) {};
		\node [style=none] (16) at (-3, 0.25) {$\propto$};
		\node [style=Z phase dot] (17) at (1.5, 0.45) {$\pi$};
		\node [style=Z phase dot] (18) at (0.5, 0.075) {$\pi$};
		\node [style=none] (19) at (4.25, 1) {};
		\node [style=none] (20) at (4.25, -0.5) {};
		\node [style=Z phase dot] (21) at (6, 1) {$\ x_i\frac{\pi}{4}+\pi~$};
		\node [style=Z phase dot] (22) at (6, -0.5) {$\ x_j\frac{\pi}{4}+\pi~$};
		\node [style=none] (23) at (7.75, 1) {};
		\node [style=none] (24) at (7.75, -0.5) {};
		\node [style=none] (25) at (3.25, 0.25) {$=$};
	\end{pgfonlayer}
	\begin{pgfonlayer}{edgelayer}
		\draw (0.center) to (2);
		\draw (1.center) to (3);
		\draw (2) to (4.center);
		\draw (5.center) to (3);
		\draw (2) to (6);
		\draw (3) to (6);
		\draw (6) to (7);
		\draw (8.center) to (10);
		\draw (9.center) to (11);
		\draw (10) to (12.center);
		\draw (13.center) to (11);
		\draw (10) to (17);
		\draw (11) to (18);
		\draw (19.center) to (21);
		\draw (20.center) to (22);
		\draw (21) to (23.center);
		\draw (24.center) to (22);
	\end{pgfonlayer}
\end{tikzpicture}
\end{equation}
The extra $\pi$ phase on the qubits is compensated by the other phases in the definition of $CZ_{\alpha}$ in Eq.~\eqref{eq:zx-gates}. A similar derivation can be done for $y_{i,j} = 0$.
We can hence assume that the diagram is written in such a way that the trivial phase gadgets are removed. There is then a connection via a phase gadget between an $x_i$ and $x_j$ pair when ${y_{i,j}\neq 0}$ and $y_{i,j} \neq 4$.

\emph{Calculating amplitudes.}---%
We first address the case of \emph{strong} simulation of an IQP circuit, i.e.~calculating amplitudes of the circuit. We will show later how we can derive \emph{weak} simulation, i.e.~sampling from strong simulation with linear overhead in the number of qubits. 
Without loss of generality, we can assume that we want to know the amplitude of observing $0^n$ from an IQP circuit $C$. We can represent $\bra{0^n}C\ket{0^n}$, up to some known power of $\sqrt{2}$, as a ZX-diagram, and simplify it as follows:
\begin{equation}\label{eq:IQP-amplitude}
    \scalebox{0.85}{\begin{tikzpicture}
	\begin{pgfonlayer}{nodelayer}
		\node [style=X dot] (0) at (-7.825, 2) {};
		\node [style=X dot] (1) at (-7.825, 0) {};
		\node [style=X dot] (2) at (-7.825, -3.5) {};
		\node [style=none] (3) at (-7.825, -1.75) {$\vdots$};
		\node [style=Z phase dot] (4) at (-6.25, 2) {$~x_1\frac{\pi}{4}\ $};
		\node [style=Z phase dot] (5) at (-6.25, 0) {$~x_2\frac{\pi}{4}\ $};
		\node [style=Z phase dot] (6) at (-6.25, -3.5) {$~x_n\frac{\pi}{4}\ $};
		\node [style=X dot] (7) at (-2, 2) {};
		\node [style=X dot] (8) at (-2, 0) {};
		\node [style=X dot] (9) at (-2, -3.5) {};
		\node [style=none] (10) at (-2, -1.75) {$\vdots$};
		\node [style=X dot] (11) at (-4.75, 1) {};
		\node [style=Z phase dot] (12) at (-3.25, 1) {$\ y_{1,2}\frac{ \pi}{4}~$};
		\node [style=X dot] (13) at (-4.75, -2.5) {};
		\node [style=Z phase dot] (14) at (-3.25, -2.5) {$\ y_{1,n}\frac{ \pi}{4}~$};
		\node [style=X dot] (15) at (-4.75, -1) {};
		\node [style=Z phase dot] (16) at (-3.25, -1) {$\ y_{2,n}\frac{ \pi}{4}~$};
		\node [style=none] (17) at (-6.25, -1.75) {$\vdots$};
		\node [style=hadamard] (18) at (-7.275, 2) {};
		\node [style=hadamard] (19) at (-7.275, 0) {};
		\node [style=hadamard] (20) at (-7.275, -3.5) {};
		\node [style=hadamard] (21) at (-2.75, 2) {};
		\node [style=hadamard] (22) at (-2.75, 0) {};
		\node [style=hadamard] (23) at (-2.75, -3.5) {};
		\node [style=Z dot] (25) at (0.25, 2) {};
		\node [style=Z dot] (26) at (0.25, 0) {};
		\node [style=Z dot] (27) at (0.25, -3.5) {};
		\node [style=none] (28) at (0.25, -1.75) {$\vdots$};
		\node [style=Z phase dot] (29) at (1.5, 2) {$~x_1\frac{\pi}{4}\ $};
		\node [style=Z phase dot] (30) at (1.5, 0) {$~x_2\frac{\pi}{4}\ $};
		\node [style=Z phase dot] (31) at (1.5, -3.5) {$~x_n\frac{\pi}{4}\ $};
		\node [style=Z dot] (32) at (5.25, 2) {};
		\node [style=Z dot] (33) at (5.25, 0) {};
		\node [style=Z dot] (34) at (5.25, -3.5) {};
		\node [style=X dot] (36) at (3, 1) {};
		\node [style=Z phase dot] (37) at (4.5, 1) {$\ y_{1,2}\frac{ \pi}{4}~$};
		\node [style=X dot] (38) at (3, -2.5) {};
		\node [style=Z phase dot] (39) at (4.5, -2.5) {$\ y_{1,n}\frac{ \pi}{4}~$};
		\node [style=X dot] (40) at (3, -1) {};
		\node [style=Z phase dot] (41) at (4.5, -1) {$\ y_{2,n}\frac{ \pi}{4}~$};
		\node [style=none] (42) at (1.5, -1.75) {$\vdots$};
		\node [style=none] (50) at (-0.75, 0) {$=$};
		\node [style=Z phase dot] (55) at (8, 2) {$~x_1\frac{\pi}{4}\ $};
		\node [style=Z phase dot] (56) at (8, 0) {$~x_2\frac{\pi}{4}\ $};
		\node [style=Z phase dot] (57) at (8, -3.5) {$~x_n\frac{\pi}{4}\ $};
		\node [style=X dot] (62) at (9.5, 1) {};
		\node [style=Z phase dot] (63) at (11, 1) {$\ y_{1,2}\frac{ \pi}{4}~$};
		\node [style=X dot] (64) at (9.5, -2.5) {};
		\node [style=Z phase dot] (65) at (11, -2.5) {$\ y_{1,n}\frac{ \pi}{4}~$};
		\node [style=X dot] (66) at (9.5, -1) {};
		\node [style=Z phase dot] (67) at (11, -1) {$\ y_{2,n}\frac{ \pi}{4}~$};
		\node [style=none] (68) at (8, -1.75) {$\vdots$};
		\node [style=none] (70) at (6.5, 0) {$=$};
	\end{pgfonlayer}
	\begin{pgfonlayer}{edgelayer}
		\draw (0) to (4);
		\draw (1) to (5);
		\draw (2) to (6);
		\draw (4) to (7);
		\draw (8) to (5);
		\draw (6) to (9);
		\draw (4) to (11);
		\draw (5) to (11);
		\draw (11) to (12);
		\draw (13) to (14);
		\draw [in=90, out=-60, looseness=0.50] (4) to (13);
		\draw (13) to (6);
		\draw (15) to (16);
		\draw (5) to (15);
		\draw (15) to (6);
		\draw (25) to (29);
		\draw (26) to (30);
		\draw (27) to (31);
		\draw (29) to (32);
		\draw (33) to (30);
		\draw (31) to (34);
		\draw (29) to (36);
		\draw (30) to (36);
		\draw (36) to (37);
		\draw (38) to (39);
		\draw [in=105, out=-60, looseness=0.75] (29) to (38);
		\draw (38) to (31);
		\draw (40) to (41);
		\draw (30) to (40);
		\draw (40) to (31);
		\draw (55) to (62);
		\draw (56) to (62);
		\draw (62) to (63);
		\draw (64) to (65);
		\draw [in=105, out=-60, looseness=0.75] (55) to (64);
		\draw (64) to (57);
		\draw (66) to (67);
		\draw (56) to (66);
		\draw (66) to (57);
	\end{pgfonlayer}
\end{tikzpicture}}
\end{equation}

To calculate the value of such diagrams, we will use a stabiliser decomposition approach~\cite{bravyi2016trading,Bravyi2019simulationofquantum}.
We will show that it is possible to remove a qubit and all its adjacent phase gadgets from a ZX-diagram at the cost of having to solve two (smaller) instances instead of one. 

The idea is to observe that the definition of the Z-spider as a linear map~\eqref{eq:def-Z-spider} means we can decompose it as a sum of diagrams containing X-spiders via~\eqref{eq:state-ZX}:
\begin{equation}\label{eq:vertex-cut}
    \begin{tikzpicture}
	\begin{pgfonlayer}{nodelayer}
		\node [style=Z phase dot] (0) at (0, 0) {$\alpha$};
		\node [style=none] (3) at (1.25, 1) {};
		\node [style=none] (4) at (1.25, -1) {};
		\node [style=none] (6) at (1, 0.25) {$\vdots$};
		\node [style=none] (7) at (2.25, 0) {$=$};
		\node [style=none] (11) at (5.5, 1) {};
		\node [style=none] (12) at (5.5, -1) {};
		\node [style=none] (14) at (4.75, 0.25) {$\vdots$};
		\node [style=X dot] (16) at (4.25, 1) {};
		\node [style=X dot] (18) at (4.25, -1) {};
		\node [style=none] (19) at (6.25, 0) {$+$};
		\node [style=none] (22) at (10.25, 1) {};
		\node [style=none] (23) at (10.25, -1) {};
		\node [style=none] (25) at (10.25, 0.25) {$\vdots$};
		\node [style=X phase dot] (27) at (9, 1) {$\pi$};
		\node [style=X phase dot] (29) at (9, -1) {$\pi$};
		\node [style=none] (30) at (7.75, 0) {$\frac{e^{i\alpha}}{\sqrt{2}^k}$};
		\node [style=none] (31) at (3.25, 0) {};
		\node [style=none] (32) at (3.375, 0) {$\frac{1}{\sqrt{2}^k}$};
	\end{pgfonlayer}
	\begin{pgfonlayer}{edgelayer}
		\draw [in=180, out=75] (0) to (3.center);
		\draw [in=-180, out=-75] (0) to (4.center);
		\draw (16) to (11.center);
		\draw (18) to (12.center);
		\draw (27) to (22.center);
		\draw (29) to (23.center);
	\end{pgfonlayer}
\end{tikzpicture} \quad \forall \alpha \in [0,2\pi]
\end{equation}
Applying this to one of the $x_i$ spiders in Eq.~\eqref{eq:IQP-amplitude}, we can then remove its previously adjacent phase gadgets, using the rules of Eq.~\eqref{eq:ZX-rules}:
\begin{equation}\label{eq:gadget-cut}
    \begin{tikzpicture}
	\begin{pgfonlayer}{nodelayer}
		\node [style=Z phase dot] (1) at (3, 0.5) {$\ x_j\frac\pi4~$};
		\node [style=X dot] (2) at (4.5, 1.25) {};
		\node [style=Z phase dot] (3) at (6, 1.25) {$\ y_{i,j}\frac\pi4~$};
		\node [style=X phase dot] (6) at (3, 2) {$a\pi$};
		\node [style=none] (8) at (3.75, -0.75) {};
		\node [style=none] (10) at (2.25, -0.75) {};
		\node [style=none] (12) at (3.05, -0.25) {$\dots$};
		\node [style=none] (13) at (7.25, 0) {=};
		\node [style=Z phase dot] (14) at (13, 0.5) {$\ (x_j+(-1)^a y_{i,j})\frac\pi4~$};
		\node [style=none] (18) at (13.75, -0.75) {};
		\node [style=none] (19) at (12.25, -0.75) {};
		\node [style=none] (20) at (13.05, -0.25) {$\dots$};
		\node [style=none] (28) at (9, 0.25) {$e^{iay_{i,j}\frac\pi4}$};
	\end{pgfonlayer}
	\begin{pgfonlayer}{edgelayer}
		\draw (1) to (2);
		\draw (2) to (3);
		\draw (6) to (2);
		\draw [bend right] (1) to (10.center);
		\draw [bend left] (1) to (8.center);
		\draw [bend right] (14) to (19.center);
		\draw [bend left] (14) to (18.center);
	\end{pgfonlayer}
\end{tikzpicture}
\end{equation}
We can view Eq.~\eqref{eq:vertex-cut} as a stabiliser decomposition, which then propagates to remove additional $T$-like phases that are adjacent. In this sense it can be seen as a special case of the stabiliser decomposition of many $\ket{\text{cat}_3}$ states connected together; see~\cite{qassim_improved_2021,kissinger2022classical}.

This cutting procedure takes linear time and creates two new diagrams representing IQP amplitudes with one less qubit. We could continue this process until nothing is left but a complex number, but this would require summing up $O(2^n)$ terms. However, it is possible to do better by removing qubits up until we are left with a fully disconnected diagram of some size $k$. This fully disconnected diagram can then be contracted in linear time as it is just the product of $k$ complex numbers. This leads to an algorithm that runs in time $O(k2^{n-k})$. Thus, it is fruitful to find a set of qubits to remove that maximizes the value of $k$, i.e.~the largest set of qubits that are not connected to each other. In the next section, we will consider how the strong simulation of IQP circuits can be represented by random graphs which will give us lower bound on the value of k.

\emph{An algorithm for random IQP circuits.}---%
We are considering two random distributions over IQP circuits.
First, the random distribution for \emph{dense} IQP circuits is obtained by uniformly and independently choosing a power of $T$ gates $x_i$ on every qubit and a random power of the $CS$ gate $y_{i,j}$ for every pair of qubits. 
These were shown in~\cite{bremner_average-case_2016}, under mild assumptions, to be hard to classically sample from in the average case.
Note that as $CS^4 = \text{id}$, in the dense case there is a 3/4 change of a non-trivial interaction between a given pair of qubits. Furthermore, since $CS^2 = CZ$, in this case the interaction is Clifford.

Second, the random distribution for \emph{$\gamma$-sparse} IQP circuits is obtained in a similar manner, but now, every pair of qubits only has a probability of $p=\gamma\frac{\ln(n)}{n}$ to have a power of a $CS$ gate between them, so that each qubit interacts with $O(\ln n)$ other qubits. It has been shown that, under slightly different hardness assumptions, for $\gamma$ large enough it is also hard to sample from these circuits~\cite{bremner_achieving_2017}.

We can define the \emph{interaction graph} of an IQP circuit as the graph where we have one vertex per qubit and where there is an edge between two vertices iff the qubits they represent are connected by a phase gadget. The maximal value of $k$ above for a given circuit then corresponds to the \emph{independence number} $\alpha(G)$ of its interaction graph~$G$. We note that finding the largest independent set of a graph can be done in $\Tilde{O}(1.1996^n)$~\cite{xiao_exact_2017}. Since the search of the maximal independent set need only to be done once and its time complexity is a lot lower than the one for the algorithm we will construct, we will omit it in the rest of the complexity analysis. 

Interestingly, interaction graphs of random IQP circuits are random graphs under the Erdős–Rényi model. Specifically, random dense $n$-qubit IQP circuits have interaction graphs distributed like $G(n,3/4)$.
This $3/4$ comes from the fact that a uniformly random power of $CS$ has a $1/4$ chance to be the identity. Similarly,  random $\gamma$-sparse IQP circuit give rise to interaction graphs distributed like $G(n,\frac{3\gamma\ln(n)}{4n})$. We can hence use tools from the random graphs literature to bound the independence number obtained, which gives us a guarantee on the time complexity of our algorithm. The following classical result will be useful in particular:
\begin{theorem}[Matula, 1972~\cite{matula_employee_1972}]
For $p\in(0,1)$, $\alpha(G(n,p))$ is tightly concentrated around $2\log_{1/(1-p)}{n}$. More precisely, let $b=\frac{1}{1-p}$, $\varepsilon>0$ and $d=2\log_{b}{n}-2\log_{b}(\log_{b}(n))+2\log_{b}{e/2}+1$, then
$$\lim_{n\rightarrow \infty} \prob{\floor{d-\varepsilon}\geq \alpha(G(n,p))\geq \floor{d+\varepsilon}} = 1.$$
\end{theorem}

\begin{corollary} \label{cor : independent set dense graph}
    Let $p\in(0,1)$, $b=\frac{1}{1-p}$ then $\alpha(G(n,p)) \geq 2\log_{b}{n}-2\log_{b}(\log_{b}(n))$ with high probability.
\end{corollary}

For random dense IQP circuits we have $p=3/4$ and hence $b=4$, so that $\log_b{n} = \frac12 \log_2 n$.
This implies that the independence number of the interaction graph is with high probability bigger than $\log_2 n-\log_2 \log_2 n$. 
Hence, our strategy for calculating an amplitude runs in $O\left((\log_2 n-\log_2\log_2 n)2^{n-\log_2{n}+\log_2\log_2n}\right) = O\left( \frac{\log^2 n}{n} 2^n\right)$.

We can derive a similar bound for $\gamma$-sparse random graphs, which we prove in the Supplemental Material.
\begin{theorem}\label{thm : independant set sparse graph}
    There exists a constant $C>0$ such that with high probability 
    $$\alpha\left(G\left(n,\frac{3\gamma\ln(n)}{4n}\right)\right) \geq C\frac{n\log\log(n)}{\log(n)}.$$
\end{theorem}
This bound implies that for random $\gamma$-sparse IQP circuits, our simulation method has a time complexity of $O\left(\frac{n\log\log(n)}{\log(n)} 2^{n\left(1- \frac{C\log\log(n)}{\log(n)}\right)} \right)$. 
Note that this bound is faster than $O(2^n/poly(n))$, for any choice of polynomial (but slower than $O(2^{cn})$ for any $c<1$).

\emph{Benchmarking.}---%
We implemented the algorithm for calculating an amplitude described above and tested it on several sizes of circuits and with different sparsities. The language used was Rust and the benchmarks ran on a single thread on a consumer laptop (Intel Core i7-10750H CPU 2.60GHz). Our results are shown in Figure~\ref{fig:Benchmarks}.
\begin{figure}[tb]
    \centering
    \hspace{-9cm}
    \includegraphics[width=\columnwidth]{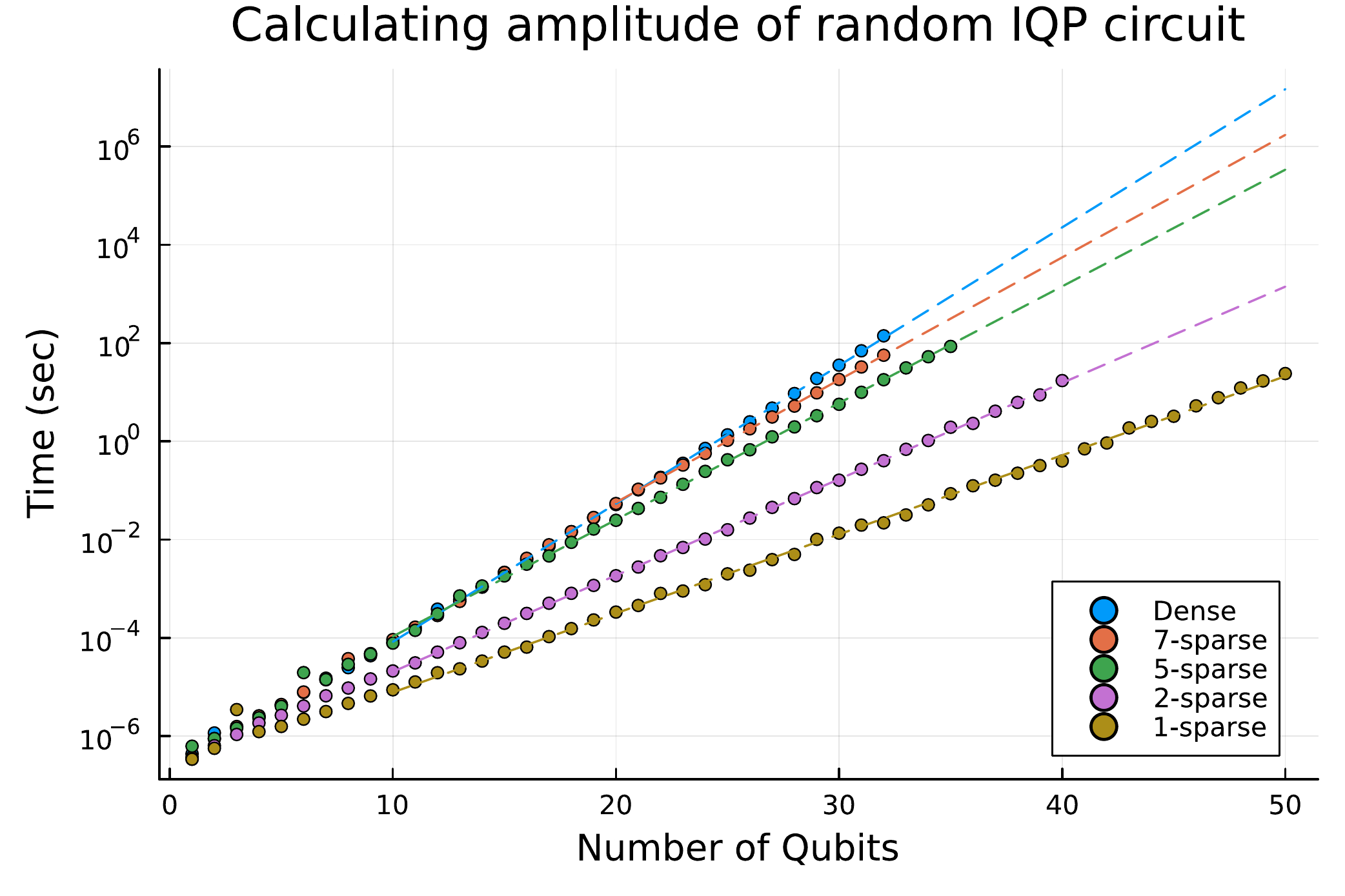}
    \vspace{-0.2cm}
    \caption{The time it takes to calculate a single amplitude using our algorithm, based on the average of $100$ instances. The dashed lines show exponential fits starting from $n=10$. Note that for 7-sparse, the graphs only start to be different from from the dense ones at $n=22$.}
    \label{fig:Benchmarks}
    \vspace{-0.3cm}
\end{figure}
We were able to calculate amplitudes from circuits with up to 50 qubits (depending on the density) in just a few minutes.
As the algorithm is easily parallelisable (since each term in the decomposition can be treated independently), we see that simulating circuits well into 60, or even 70, qubits should be possible with a sizable computing cluster.

Note that the data fits remarkably well to an exponential fit $c2^{\alpha n}$ where $\alpha$ ranges from $0.93$ to $0.53$. This suggests that a better upper bound on the complexity of simulating sparse circuits might be possible.


\emph{Weak simulation.}---The above only describes how to calculate amplitudes of IQP circuits. To sample from circuits, we can use the strong simulation procedure described above as a subroutine of the `gate-by-gate' simulation technique of~\cite{bravyi2021simulate} that avoids calculating marginals. This technique requires the computation of an amplitude for every non-diagonal gate in the circuit. Each of these amplitudes is based on a subcircuit of the original. We note that calculating such an amplitude is at most as hard as calculating an amplitude of the full circuit (and in fact, will often be much easier). Since IQP circuits only have $2n$ non-diagonal gates (corresponding to the layers of Hadamard gates), bootstrapping our strong simulation algorithm to a weak simulation one only adds a linear overhead $O(n)$, which in practice can be negligible.

\emph{An improved algorithm for calculating amplitudes.}---%
It is possible to pick a different set of qubits to decompose with Eq.~\eqref{eq:vertex-cut} which leads to a better asymptotic complexity in the dense case, and might also give practical benefit in the sparse setting.
The idea is to stop cutting vertices before completely disconnecting the diagram and then use a general stabiliser decomposition algorithm. This two-step process allows us to bring down the number of $T$ gates from $O(n^2)$ to a more manageable $O(n)$ before using a more efficient stabiliser decomposition algorithm. To do so, let's consider the \textit{non-Clifford interaction graph} of an IQP circuit. In this graph, there is an edge between two vertices only if they are connected by a non-Clifford phase gadget (i.e.~when $y_{i,j}$ in Eq.~\eqref{eq:IQP-amplitude} is odd). Finding the largest independent set of this graph and removing all the other qubits using Eqs.~\eqref{eq:vertex-cut} and~\eqref{eq:gadget-cut} then results in a diagram where all the interactions between two qubits are Clifford. This diagram can then be given to a stabiliser decomposition algorithm such as that in~\cite{kissinger2022classical}. On average, half of the qubits in the diagram will have a non-Clifford phase that comes from the initial layer of powers of $T$ gates in the construction of the IQP circuit. Therefore, this algorithm runs in time $O\left(2^{n-k}f(k/2)\right)$ where $k$ is the size of the largest independent set and $f(k/2)$ is the time taken to calculate the amplitude of a diagram with $k/2$ $T$ gates by a dedicated stabiliser decomposition algorithm. At the time of this writing, the best general-purpose stabiliser decomposition algorithm is from~\cite{kissinger2022classical} and has a time complexity of $f(t) = O\left(t^2 2^{\beta t}\right)$ where $\beta = \log_2(3)/4 \approx 0.396$.

The main advantage of using this modified approach is that the non-Clifford interaction graph is less dense than the standard interaction graphs, while still being Erdős–Rényi random.
More precisely, the non-Clifford interaction graph is distributed as $G(n,1/2)$. By Corollary~\ref{cor : independent set dense graph}, this graph has with high probability an independent set of size $2\log_2{n}-2\log_2\log_2 n$. Using this approach, calculating an amplitude then runs in $O\left(\frac{(\log n)^{4-\beta}}{n^{2-\beta}} 2^n \right)$. Since $2-\beta \approx 1.604$, this is an improvement over our first approach.
Using this improved algorithm for sparse circuits results in the same asymptotic complexity as we found before, but might still be better in practice. But on the other hand, this method does introduce significant complexity in the implementation which might in fact result in enough slow down to cancel out the asymptotic benefit for relevant parameters.

\emph{Conclusion.}---%
We found a new algorithm for exactly calculating amplitudes of random IQP circuits that both in the dense and sparse setting improve upon the previous asymptotic complexity and allow us to simulate large circuits in practice. Our results show that current and near-term hardware is probably not yet at a level where quantum supremacy could definitively be shown using random IQP circuit sampling.
Our benchmarks suggest that it might be possible to derive better asymptotic bounds for the cost of simulating sparse circuits.

\emph{Acknowledgements.}---%
We would like to thank Tuomas Laakkonen for his fruitful comments and his help with the implementation of the algorithm. We would also like to thank Oliver Riordan for his help in the analysis of the independence number of $\gamma$-sparse random graphs. We acknowledge the support of the Natural Sciences and Engineering Research Council of Canada (NSERC).

\bibliography{main}

\appendix
\section{Independence number of \texorpdfstring{$\gamma$}{gamma}-sparse random graphs}
\label{section: independence proof}

We here restate theorem~\ref{thm : independant set sparse graph}: \\
\noindent \textbf{Theorem \ref{thm : independant set sparse graph}.}
\textit{There exist a constant $C>0$ such that with high probability $$\alpha\left(G\left(n,\frac{3\gamma\ln(n)}{4n}\right)\right) \geq C\frac{n\log\log(n)}{\log(n)}.$$}

The idea behind this proof is to use a classic result from Shearer \cite{shearer_note_1983} about independent sets in triangle-free graphs.
\begin{theorem}[Shearer 1983]\label{thm: density independant}
Let $G$ be a triangle-free graph on n points with average degree $d$, then 
    \begin{equation}
        \alpha(G) \geq n (d \ln(d) - d + 1)/(d - 1)^2.
    \end{equation}
\end{theorem}

Even though $\gamma$-sparse random graphs aren't triangle-free, with high probability they contain rather few triangles. We utilise this fact by removing vertices from $G$ until it is triangle-free.
\begin{lemma}
    With high probability $G\left(n,\frac{3\gamma\ln n }{4n}\right)$ has less than $\ln^4 n$ triangles.
\end{lemma}
\begin{proof}
    Let $X$ be the random variable representing the number of triangles in $G\left(n,\frac{3\gamma\ln n}{4n}\right)$. Then 
    \begin{align*}
        \frac{\esp{X}}{\ln^4 n} &=  \frac{1}{\ln^4 n} \binom{n}{3} \left(\frac{3\gamma\ln n}{4n}\right)^3\\
        &=  \frac{O(\ln^3(n))}{\ln^4(n)}\\
        &=o(1)
    \end{align*} 
    Applying Markov's inequality gives us the result.
\end{proof}
By removing one vertex per triangle of $G\left(n,\frac{3\gamma\ln n}{4n}\right)$, we obtain a triangle free graph $G'$. By the lemma, $G'$ contains with high probability more than $n-\ln^4 n$ vertices.
Let us denote the number of vertices and the average degree of $G'$ by $n'$ and $d'$ respectively. 
Assume that $n$ is large enough that $\ln^4 n < \frac12 n$, so that $n' > \frac12 n$ with high probability.
Notice that $d'\leq \frac{2|E(G)|}{n'}$. By Chernoff bound, $|E(G)| \leq \gamma n\ln(n) $ with high probability. Therefore, $d'\leq 2\gamma n\ln(n)/ n' \leq 4\gamma n \ln(n)/n = 4\gamma \ln n$ with high probability.
We can now use theorem \ref{thm: density independant} to prove the theorem.

\begin{proof} (of Theorem \ref{thm : independant set sparse graph})
We calculate:
\begin{align*}
    \alpha(G) &\geq \alpha(G')\\
            &\geq n' \frac{d' \ln(d') - d' + 1}{(d' - 1)^2}\\
            &= O\left(n' \frac{d' \ln(d')}{d'^2}\right)\\
            &= O\left(n' \frac{\ln(d')}{d'}\right)\\
            &\geq O\left(n' \frac{\ln(4\gamma\ln(n))}{4\gamma\ln(n)}\right)\\
            &\geq O\left((n-\ln^4(n)) \frac{\ln(4\gamma\ln(n))}{4\gamma\ln(n)}\right)\\
            &\geq O\left(n \frac{\ln(\ln(n))}{\ln(n)}\right)
\end{align*}
Hence, $\alpha(G) \geq C n \frac{\ln(\ln(n))}{\ln(n)}$ for some $C>0$ (with high probability) when $n$ is large enough.
\end{proof}

We note that the bounds used to derive this theorem are quite crude when $n$ is small.
For the circuit sizes we considered in our benchmarks, the independent sets were much larger than one could expect by simply looking at those asymptotic results.

\end{document}